\documentclass[
final
]{dmtcs-episciences}

\usepackage[utf8]{inputenc}
\usepackage{subfigure}

\usepackage[round]{natbib}

\usepackage[margin=1.1in]{geometry}
\usepackage{amsmath,amssymb,amsthm}
\usepackage{enumerate}
\usepackage[shortlabels]{enumitem}
\usepackage[usenames,dvipsnames,svgnames,table]{xcolor}
\usepackage{tikz}
\usepackage{tikz-cd}
\usepackage{array}
\usetikzlibrary{arrows,shapes,shapes.misc,positioning,calc}
\usepackage{ifthen}
\usepackage{multirow}
\usepackage{hyperref}
\usepackage[capitalise]{cleveref}
\usepackage{xfrac}
\usepackage{MnSymbol}
\usepackage{bbold}
\usepackage[normalem]{ulem}

\newtheorem{proposition}{Proposition}[section]
\newtheorem{theorem}[proposition]{Theorem}

\newtheorem{lemma}[proposition]{Lemma}
\newtheorem{corollary}[proposition]{Corollary}
\theoremstyle{definition}
\newtheorem{definition}[proposition]{Definition}
\newtheorem{remark}[proposition]{Remark}
\newtheorem{example}[proposition]{Example}

\newcommand{\B}{\mathbb{B}}
\newcommand{\tf}{\tilde{f}}
\newcommand{\tA}{\tilde{A}}
\newcommand{\tV}{\tilde{V}}
\renewcommand{\S}{\mathcal{S}}
\newcommand{\R}{\mathcal{R}}
\newcommand{\0}{\mathbb{0}}
\newcommand{\1}{\mathbb{1}}
\newcommand*\smbox[1]{\tikz[baseline=(char.base)]{\node[shape=rectangle,draw,inner sep=1.5pt] (char) {$#1$};}}

\author[Robert Schwieger and Elisa Tonello]{Robert Schwieger
  \and Elisa Tonello\affiliationmark{1}\thanks{Funded by the Deutsche Forschungsgemeinschaft (DFG, German Research Foundation) under Germany's Excellence Strategy – The Berlin Mathematics Research Center MATH+ (EXC-2046/1, project ID: 390685689).}}
\title[Reduction for asynchronous Boolean networks]{Reduction for asynchronous Boolean networks:
  elimination of negatively autoregulated components}
\affiliation{
  Department of Mathematics and Computer Science, Freie Universit\"at Berlin, Berlin, Germany}
\begin{document}
\publicationdata{vol. 25:2}{2023}{27}{10.46298/dmtcs.10930}{2023-02-08}{2023-11-28}
\maketitle
\begin{abstract}
  To simplify the analysis of Boolean networks, a reduction in the number of components is often considered. A popular reduction method consists in eliminating components that are not autoregulated, using variable substitution.
In this work, we show how this method can be extended, for asynchronous dynamics of Boolean networks, to the elimination of vertices that have a negative autoregulation, and study the effects on the dynamics and interaction structure.
For elimination of non-autoregulated variables, the preservation of attractors is in general guaranteed only for fixed points. Here we give sufficient conditions for the preservation of complex attractors. The removal of so called mediator nodes (i.e. vertices with indegree and outdegree one) is often considered, and frequently does not affect the attractor landscape. We clarify that this is not always the case, and in some situations even subtle changes in the interaction structure can lead to a different asymptotic behaviour.
Finally, we use properties of the more general elimination method introduced here to give an alternative proof for a bound on the number of attractors of asynchronous Boolean networks in terms of the cardinality of positive feedback vertex sets of the interaction graph.

\end{abstract}

\section{Introduction}

With increasingly powerful technologies in molecular genetics it is possible to obtain large amount of data which lead to increasingly larger models of complex regulatory networks. This poses problems and limitations on the analysis of such models. While this applies especially to quantitative models (e.g., differential or stochastic models \cite{karlebach2008modelling,radulescu2012reduction,saunders2013coarse}), qualitative models are also increasingly affected. Among the latter, logical models are widely used \cite{abou2016logical,samaga2013modeling,albert2014boolean,le2015quantitative}.

Despite their simplicity, the combinatorial explosion with the increasing number of components makes the rigorous analysis of many models unattainable.
An approach to deal with this problem consists in reducing the size of the original network. There are mainly two strategies in use. The first one relies on trap spaces, i.e. invariant subspaces of the state space \cite{Klarner2015}.
The second approach, on which we will focus here, relies on the assumption that some of the updates, i.e. changes in the components, are happening faster than others. This idea has been developed for the Boolean as well as for the more general multi-valued case \cite{naldi2009reduction,naldi2011dynamically,veliz2011reduction}.
The method allows, in the Boolean formalism, to substitute a variable with the expression defining its update rule. This approach is only possible if the variable is not autoregulated. In terms of asynchronous state transition graphs, the absence of autoregulation guarantees that, for each pair of neighbour states that differ in the variable being eliminated, exactly one of the two states is the source of a transition that changes the value of the variable being eliminated. The other state is therefore the target of this transition, and can be selected as the ``representative'' state; all transitions from representative states are preserved by the elimination. In this setting, we observe that there is a natural way of extending the elimination to variables that are negatively autoregulated. In presence of a possible negative autoregulation, a pair of neighbour states that differ in the variable being eliminated can be connected by transitions in both directions. In this case it is not necessary to choose a representative, and since the two states are part of the same strongly connected component, transitions from any of the two states can be preserved in the reduction.
The elimination method introduced in this work implements this idea.
We show that this extended method affects the interaction graph in a similar way to the original reduction method, with some differences that can concern the introduction of loops. While the preservation of fixed points needs to be refined to account for attractors consisting of two states that can collapse to one, we prove that the total number of attractors cannot decrease with the reduction, as for the original method.
Using these properties, we give an alternative proof for a result, due to Richard~\cite{richard2009positive}, that establishes a bound on the number of attractors of asynchronous Boolean networks in terms of the cardinality of positive feedback vertex sets of the interaction graphs.

The reduced networks of the method introduced in~\cite{naldi2009reduction,naldi2011dynamically,veliz2011reduction} can be computed quite easily, making the approach applicable to very large networks.
While fixed points are always preserved by the elimination of variables that are not autoregulating,
in some cases this reduction approach can change the dynamics of the networks significantly.
Therefore, some effort has been invested in finding conditions on the structure of the network for which it can be guaranteed that a reduction not only preserves fixed points, but all attractors. In \cite{saadatpour2010attractor,saadatpour2013reduction}, the authors suggested the merging of vertices which have in- and outdegree one, so-called simple mediator nodes \cite{saadatpour2013reduction} (also called linear variables in \cite{naldi2023linear}). Here we take a detailed look at these assumptions and show that there are unfortunately still certain cases where attractors are not preserved, despite the claim in \cite{saadatpour2013reduction}. This result does not impact the usefulness of the method suggested in \cite{saadatpour2010attractor,saadatpour2013reduction} since such counterexamples can be quite artificial in nature.

In \cref{sec:background} we set the required notation and give a brief summary of some properties of the reduction method described in~\cite{naldi2009reduction,naldi2011dynamically,veliz2011reduction}. We then introduce a generalisation of this reduction method that can be applied to variables with negative autoregulation, and use it to derive a simple proof for a bound on the number of attractors of Boolean networks (\cref{sec:generalisation-neg-loops}).
In~\cref{sec:attractor-preservation} we discuss the preservation of cyclic attractors under elimination of intermediate components, with or without negative autoregulation.

\section{Background and notation}\label{sec:background}

A Boolean network is a map $f\colon\B^n\to\B^n$, where $\B=\{0,1\}$.
We call $V=\{1,\dots,n\}$ the set of components of the Boolean network, and $\B^n$ the set of states.
Given $x\in\B^n$ and $I\subseteq V$, we denote by $\bar{x}^I$ the element of $\B^n$ such that $\bar{x}^I_i=1-x_i$ for $i\in I$
and $\bar{x}^I_i=x_i$ for $i\notin I$.
We write $\bar{x}$ for $\bar{x}^{V}$, and, given $i\in V$, we write $x^i$ for $x^{\{i\}}$.
In addition, given $a\in\B$, $x^{i=a}$ denotes the element of $\B^n$ obtained from $x$
by setting the $i^{th}$ component to $a$.

This work deals with elimination of variables. From a Boolean network $f\colon\B^n\to\B^n$ with $n$ components we will define a Boolean network $\tf\colon\B^{n-1}\to\B^{n-1}$ with $n-1$ components. To simplify the notation, after removing variable $v$ we will use the indices $\tV=\{1,\dots,v-1,v+1,\dots,n\}$ to identify components of the Boolean network $\tf$ and of states in $\B^{n-1}$. We will write $\pi\colon\B^n\to\B^{n-1}$ for the projection onto the components $\tV$.

The \emph{asynchronous dynamics} or \emph{asynchronous state transition graph} $AD(f)$ associated to a Boolean network $f$ with set of components $V$ is a directed graph with
set of vertices or \emph{states} $\B^n$, and set of edges or \emph{transitions} defined by $\{(x,\bar{x}^i)\ |\ i\in V, f_i(x)\neq x_i\}$.
The asynchronous dynamics is frequently considered when modelling gene regulatory networks~\cite{samaga2013modeling,le2015quantitative,abou2016logical}.

Given $x\in\B^n$, the \emph{(local) interaction graph} $G(f)(x)$ of $f$ at $x$ is the signed directed graph
with set of vertices $V$ and admitting an edge from $j$ to $i$ of sign $s\in\{-1,1\}$ if and only if
$s = (f_i(\bar{x}^j)-f_i(x))(\bar{x}^j_j-x_j)$.
The \emph{(global) interaction graph} $G(f)$ of $f$ is the union of the local interaction graphs, i.e. the signed multidirected graph with set of vertices $V$ and set of edges given by the union of the edges in $G(f)(x)$ for all $x\in\B^n$.
If $G$ has an edge from $j$ to $i$, then $j$ is said to be a \emph{regulator} of $i$. A loop in $G$, that is, an edge of the form $(i,i)$ is also called an \emph{autoregulation} of the variable $i$.
The interaction graph is used to summarize the relationships between variables.
Its features can often be related to properties of state transition graphs (see e.g. \cite{pauleve2012static,comet2013circuit,richard2019positive}).

Edges in state transition graphs and interaction graphs will be denoted with arrows (e.g., $x\to y$ for the edge $(x,y)$).
A path in a directed graph $G$ is defined by a sequence of edges $x^1 \to x^2 \to \dots \to x^{k-1} \to x^k$.
We call the number of edges defining the path the \emph{length} of the path, and the vertices in the path the \emph{support} of the path.
If the edges are signed, we define the sign of the path as the product of the signs of its edges.
If all vertices in the path are distinct, with the possible exception of the first and the last vertices,
we say that the path is \emph{elementary}.
If the first and the last vertices in an elementary path coincide, we call the path a \emph{cycle}.
If a path is a cycle of length one, it will also be called a \emph{loop}.

A \emph{trap set} is a subset $T$ of $\B^n$ such that, for any $x\in T$ and $x\to y$ transition in $AD(f)$, $y$ is in $T$.
The minimal trap sets are call the \emph{attractors} of $AD(f)$. Attractors are called \emph{fixed points} if they
contain only one state, \emph{cyclic attractors} otherwise.

\subsection{Elimination of non-autoregulated components}\label{sec:elim-classical}

A reduction method has been introduced for Boolean and more general discrete networks~\cite{naldi2009reduction,naldi2011dynamically,veliz2011reduction},
which allows to eliminate components that do not admit loops in the interaction graph.
The method has been extensively applied~\cite{calzone2010mathematical,saadatpour2010attractor,saadatpour2011dynamical,grieco2013integrative,paracha2014formal,quinones2014dynamical,zanudo2015cell,flobak2015discovery}.
In the first part of this work we investigate the elimination of components that admit negative autoregulation, in the Boolean case.
The approach provides an extension of the original method, and opens new venues for application.
Before introducing our extension, in this section we summarize some properties of the method introduced in~\cite{naldi2009reduction,naldi2011dynamically} and in the Boolean case in~\cite{veliz2011reduction}, to ease the introduction of the new approach.

Consider a Boolean network $f\colon\B^n\to\B^n$ and a vertex $v\in V$ such that there is no loop at $v$ in $G(f)$,
that is, $f_v(x)=f_v(\bar{x}^v)$ for all $x\in\B^n$.
Define the map
\begin{equation*}
  \begin{aligned}
    \R\colon\B^n&\to\B^n,\\
    x&\mapsto (x_1,\dots,x_{v-1},f_v(x),x_{v+1},\dots,x_n).
  \end{aligned}
\end{equation*}
We say that $\R$ maps each state to its \emph{representative state} in $\B^n$.
The absence of loops at $v$ in $G(f)$ implies that $\R(x)=\R(\bar{x}^v)$ for each $x\in\B^n$, and consequently there are exactly $2^{n-1}$ representative states.
For simplicity, denote by $\pi\colon\B^n\to\B^{n-1}$ the projection onto the variables $V\setminus\{v\}$.
Since $\R(x)=\R(y)$ for all $x,y\in\B^n$ for which $\pi(x)=\pi(y)$ holds,
there is a unique map $\S\colon\B^{n-1}\to\B^n$ that satisfies $\S\circ\pi=\R$ (see~\cref{fig:def-reduction} left).

\begin{figure}
\centering
\begin{tikzcd}
\B^n \arrow[r, "\R"] \arrow[d, "\pi"] & \B^n \\
\B^{n-1} \arrow[ru, "\S", dashed]          &  
\end{tikzcd}\hspace{50pt}
\begin{tikzcd}
\B^n \arrow[r, "f"]           & \B^n \arrow[d, "\pi"] \\
\B^{n-1} \arrow[r, "\tf", dashed] \arrow[u, "\S"] & \B^{n-1}          
\end{tikzcd}
\caption{Commutative diagrams that illustrate the definition of the reduction method described in~\cite{naldi2009reduction,naldi2011dynamically,veliz2011reduction}.}\label{fig:def-reduction}
\end{figure}
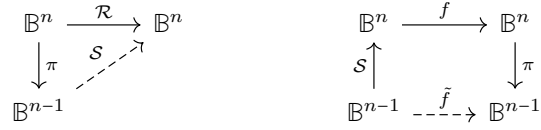

We can then define the reduced Boolean network $\tf\colon\B^{n-1}\to\B^{n-1}$ as follows (see~\cref{fig:def-reduction} right):
\begin{equation}\label{eq:reduction-original}
  \begin{aligned}
    \tf=\pi\circ f\circ\S.
  \end{aligned}
\end{equation}

The effect of the elimination on the asynchronous dynamics is represented in~\cref{fig:idea-of-elim} (left).
For convenience, as mentioned in the background, we use the set $V\setminus\{v\}$ to index the components of $\tf$ and of states in $\B^{n-1}$.

We give a small example for illustration.

\begin{example} Consider the Boolean network $f$ defined as

\begin{equation*}
  \begin{aligned}
    f\colon\B^3&\to\B^3,\\
    x&\mapsto ((\bar{x}_2 \wedge x_3) \vee (x_2 \wedge \bar{x}_3),
               (x_1 \wedge x_3) \vee (\bar{x}_1 \wedge \bar{x}_3),
               (\bar{x}_1 \wedge \bar{x}_2) \vee (x_2 \wedge x_3)),
  \end{aligned}
\end{equation*}

Its state transition graph is depicted in~\cref{fig:small-example-for-general-reduction} left. We remove variable $x_2$. Thus, in the above terminology $\pi$ is the projection onto the first and third component, $\S$ is given by $\S\colon\B^{2}\to\B^3, (x_1, x_3) \mapsto (x_1, (x_1 \wedge x_3) \vee (\bar{x}_1 \wedge \bar{x}_3) ,x_3)$ and $\R$ maps $(x_1,x_2,x_3)$ to $(x_1,(x_1 \wedge x_3) \vee (\bar{x}_1 \wedge \bar{x}_3),x_3)$. The representative states $001$, $010$, $100$ and $111$ are represented in boxes in \cref{fig:small-example-for-general-reduction}. To remove $x_2$ we substitute $x_2$ with $(x_1 \wedge x_3) \vee (\bar{x}_1 \wedge \bar{x}_3)$ in $f_1$ and $f_3$. We obtain:

\begin{equation*}
  \begin{aligned}
    \tf\colon\B^2&\to\B^2,\\
                x&\mapsto (\bar{x}_1, x_3).
  \end{aligned}
\end{equation*}

\begin{figure}
\centering
 \begin{tikzcd}[column sep=tiny,row sep=tiny]
   & 011 \arrow[dl] \arrow[from=dd,to=ddrr] & & \smbox{111} \arrow[ll] \\
   \smbox{001} \arrow[rr,crossing over] & & 101 \arrow[ru] \arrow[dd,crossing over] & \\
   & \smbox{010} & & {\color{black}110} \arrow[dl] \\
   {\color{black}000} \arrow[uu] \arrow[ur] & & \smbox{100} \arrow[ll] &
 \end{tikzcd}\qquad
 \begin{tikzcd}
   01 \arrow[r,yshift=+1pt] & 11 \arrow[l,yshift=-1pt] \\
   00 \arrow[r,yshift=+1pt] & 10 \arrow[l,yshift=-1pt]
 \end{tikzcd}\caption{Illustration of the reduction method described in~\cite{naldi2009reduction}. Representative states are shown in boxes. When the second variable is eliminated, transitions starting from representative states are preserved. The asynchronous dynamics on $\B^3$ on the left reduces to the asynchronous dynamics on $\B^2$ on the right.}\label{fig:small-example-for-general-reduction}
\end{figure}
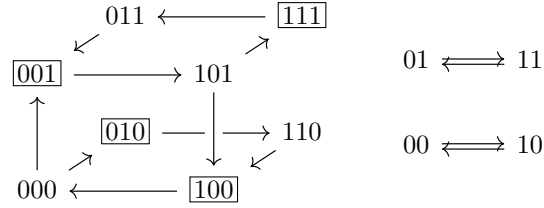

 The state transition graph of the reduced network is represented in~\cref{fig:small-example-for-general-reduction} right.
\end{example}

In the above example we see that some edges ``disappear'' during the reduction. For example there is an edge from $101$ to $100$ in the state transition graph of the original Boolean network while there is no edge from $\pi(101)=11$ to $\pi(100)=10$ in the reduced one. On the other hand, the outgoing edges from the representative states can be found also in the reduced network.
The following results can be found, with slightly different statements, in~\cite{naldi2009reduction,naldi2011dynamically}.
We will prove generalizations of these results in the next section.

\begin{proposition}\label{lemma}
  Consider a Boolean network $f\colon\B^n\to\B^n$ such that there is no loop at $v$ in $G(f)$.
  \begin{enumerate}[label=(\roman*)]
    \item\label{lem:x-to-repr} For each $x\in\B^n$, if $x\neq\R(x)$ there is a transition in $AD(f)$ from $x$ to $\R(x)$.
    \item\label{lem:repr-to-any} For all $x,y\in\B^n$, if $\R(x)\to y$ is a transition in $AD(f)$,
      then $\pi(x)=\pi(\R(x))\to\pi(y)$ is a transition in $AD(\tf)$.
    \item\label{lem:not-infl} For each $x\in\B^n$ and $i\in V \setminus\{v\}$ such that there is no edge $v\to i$ in $G(f)$,
      if $x\to\bar{x}^i$ is a transition in $AD(f)$, then $\R(x)\to\overline{\R(x)}^i$ is a transition in $AD(f)$
      and $\pi(x)\to\pi(\bar{x}^i)$ is a transition in $AD(\tf)$.
  \end{enumerate}
\end{proposition}

Even though, in general, the number of attractors can change during the reduction, the number of fixed points remains the same.

\begin{theorem}\label{thm:prev-results}
  Consider a Boolean network $f\colon\B^n\to\B^n$ such that there is no loop at $v$ in $G(f)$.
  \begin{enumerate}[label=(\roman*)]
    \item\label{thm:fixed1} If $x\in\B^n$ is a fixed point for $f$, then $x=\R(x)$, $\pi(x)$ is a fixed point for $\tf$ and no other
                          fixed point for $f$ is projected on $\pi(x)$.
    \item\label{thm:fixed2} If $x\in\B^{n-1}$ is a fixed point for $\tf$, then $\S(x)$ is a fixed point for $f$.
    \item\label{thm:trapsets} If $T\subseteq\B^n$ is a trap set for $f$, then $\pi(T)$ is a trap set for $\tf$.
    \item\label{thm:attrft} If $\tA\subseteq\B^{n-1}$ is a cyclic attractor for $\tf$, there exists at most one attractor for $AD(f)$ intersecting $\pi^{-1}(\tA)$.
  \end{enumerate}
\end{theorem}

\section{Generalisation to vertices with optional negative autoregulation}\label{sec:generalisation-neg-loops}

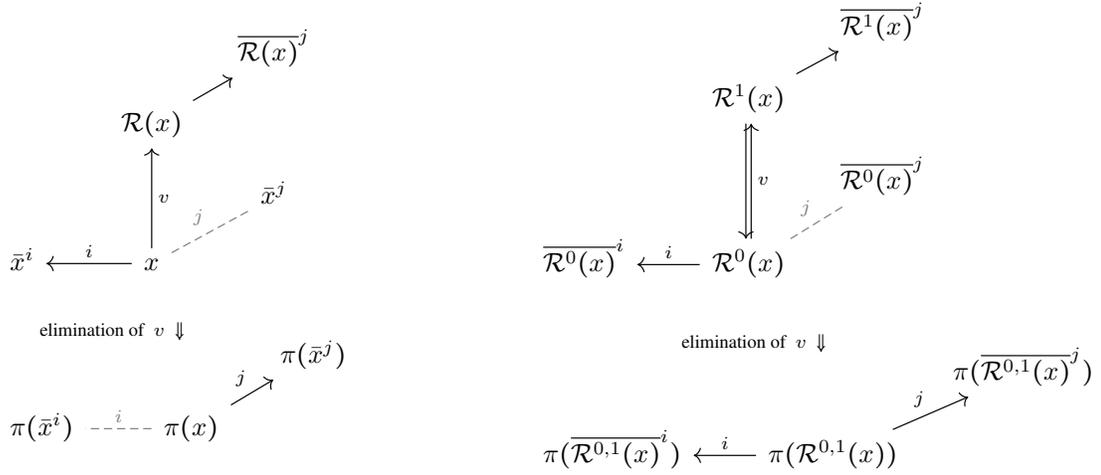
\begin{figure}
\centering
\begin{minipage}{7cm}
  \begin{tikzcd}[row sep=small,column sep=small]
    & & & \overline{\R(x)}^j \\
    & & \R(x) \arrow[ur] & \\
    & & & \bar{x}^j \\
    \bar{x}^i & & x \arrow[uu,"v"'] \arrow[ll,"i"'] \arrow[ur,dash,dashed,gray,"j"] &
  \end{tikzcd}

  \vspace{-40pt}
  \begin{tikzcd}[row sep=small,column sep=small]
    & & & {\color{white}\overline{\R(x)}^j} \\
    & & {\color{white}\R(x)} \arrow[white]{dd}[black,near start,swap]{\Large \text{elimination of}\ v\ \Downarrow} & \\
    & & & \pi(\bar{x}^j) \\
    \pi(\bar{x}^i) & & \pi(x) \arrow[ll,dash,dashed,gray,"i"'] \arrow[ur,"j"] &
  \end{tikzcd}
\end{minipage}
\begin{minipage}{7cm}
  \begin{tikzcd}[row sep=small,column sep=small]
    & & & \overline{\R^1(x)}^j \\
    & & \R^1(x) \arrow[ur] \arrow[dd,xshift=-1pt] & \\
    & & & \overline{\R^0(x)}^j \\
    \overline{\R^0(x)}^i & & \R^0(x) \arrow[uu,xshift=1pt,"v"'] \arrow[ll,"i"'] \arrow[ur,dash,dashed,gray,"j"] &
  \end{tikzcd}

  \vspace{-40pt}
  \begin{tikzcd}[row sep=small,column sep=small]
    & & & {\color{white}\overline{\R(x)}^j} \\
    & & {\color{white}\R(x)} \arrow[white]{dd}[black,near start,swap]{\Large \text{elimination of}\ v\ \Downarrow} & \\
    & & & \pi(\overline{\R^{0,1}(x)}^j) \\
    \pi(\overline{\R^{0,1}(x)}^i) & & \pi(\R^{0,1}(x)) \arrow[ll,"i"'] \arrow[ur,"j"] &
  \end{tikzcd}
\end{minipage}
\caption{Illustration of the effect of elimination of one variable ($v$)
on asynchronous state transition graphs
in case of no loops in $G(f)(x)$ at $v$ (left) and a negative loop in $G(f)(x)$ at $v$ (right).
In the first case, only transitions that start at the representative state $\R(x)$ are preserved.
In the second case, transitions out of both $\R^0(x)$ and $\R^1(x)$ are preserved.
}\label{fig:idea-of-elim}
\end{figure}

The goal of this section is to generalize the elimination method from the last section to variables with negative autoregulation. The method summarised in \cref{sec:elim-classical} applies to the elimination of variables which are not autoregulated. In this case the eliminated variable is replaced by its update function.
To generalize this idea, we substitute a variable with a more complicated expression derived from its update function. If there is no autoregulation in a state, this expression coincides with the update function of the variable. If there is a negative autoregulation, the variable we want to remove oscillates at some state. In this case, the expression is constructed in such a way that all transitions originating from this state, as well as from its neighbouring state differing only in the eliminated component, are retained in the reduced network.

Fix a Boolean network $f\colon\B^n\to\B^n$ and a vertex $v \in V$ such that there is no positive loop at $v$ in $G(f)$. Since $v$ is potentially autoregulated, the definition of representative state of the previous section cannot be applied.
Since the value of component $v$ might oscillate, we have to introduce two new functions, the maps

\begin{equation*}
  \begin{aligned}
    \R^a\colon\B^n&\to\B^n,\\
    x&\mapsto (x_1,\dots,x_{v-1},f_v(x^{v=a}),x_{v+1},\dots,x_n),
  \end{aligned}
\end{equation*}

for $a \in \{0,1\}$.
For $x\in \B^n$, $\R^0(x)$ and $\R^1(x)$ differ if $f_v(x^{v=0})\neq f_v(x^{v=1})$, that is, if component $v$ is autoregulated at $x$.

\begin{remark}\label{rmk:R0noteqR1}
  Observe that, if $\R^0(x)\neq\R^1(x)$, since $v$ is not positively autoregulated, we have that $v$ is negatively autoregulated at $x$,
  $f_v(x^{v=0})=1$ and $f_v(x^{v=1})=0$, and therefore $\R^0(x)=x^{v=1}$, $\R^1(x)=x^{v=0}$.
\end{remark}

Clearly we have $\pi(x)=\pi(\R^0(x))=\pi(\R^1(x))$, and, similarly to the case of the previous section, there are two unique maps $\S^0,\S^1\colon\B^{n-1}\to\B^n$
that satisfy $\R^0=\S^0\circ\pi$, $\R^1=\S^1\circ\pi$.

We can now introduce the reduced Boolean network $\tf\colon\B^{n-1}\to\B^{n-1}$ defined by
\begin{equation}\label{eq:casesftilde}
\tf_i(x)=\begin{cases}
           f_i(\S^0(x))\wedge f_i(\S^1(x)) & \text{ if }x_i=1,\\
           f_i(\S^0(x))\vee f_i(\S^1(x)) & \text{ if }x_i=0.
         \end{cases}
\end{equation}

More compactly, we can write that $\tf_i(x) = x_i$ if and only if $f_i(\S^0(x)) = f_i(\S^1(x)) = x_i$.

Observe that if $v$ is not autoregulated the equalities $\R^0=\R^1$, $\S^0=\S^1$ hold and therefore in this case the definition of $\tf$ coincides with the definition of $\tf$ in~\cref{eq:reduction-original}. In other words, the above reduction method is a generalization of the reduction method reviewed in the last section. 

If $v$ is instead autoregulated at $x$ and $\R^0(x) \not = \R^1(x)$, the intuition is that all the outgoing edges of both $\R^0(x)$ and $\R^1(x)$ along the components $V \setminus \{v\}$ are preserved in the reduced network. \cref{fig:idea-of-elim} (right) gives an illustration of this idea.

\begin{lemma}\label{lemma2}
  Consider a Boolean network $f\colon\B^n\to\B^n$ such that there is no positive loop at $v$ in $G(f)$. Then:
  \begin{enumerate}[label=(\roman*)]
    \item\label{lem:x-to-repr2} For each $x\in\B^n$ such that $x\neq\R^0(x)$ or $x\neq\R^1(x)$ there is a transition in $AD(f)$ from $x$ to $\R^0(x)$ or to $\R^1(x)$
      and $\{\R^0(x),\R^1(x)\}$ is strongly connected.
    \item\label{lem:repr-to-any2} For all $x\in\B^n$ and $a=0,1$, if $\R^a(x)\to y$ is a transition in $AD(f)$ in direction $i \in V \setminus \{v\}$,
      then $\pi(x)=\pi(\R^a(x))\to\pi(\overline{x}^i)$ is a transition in $AD(\tf)$.
    \item\label{lem:not-infl2} For each $x\in\B^n$ and $i\in V\setminus\{v\}$ such that there is no edge from $v$ to $i$ in $G(f)$,
      if $x\to\bar{x}^i$ is a transition in $AD(f)$, then $\R^a(x)\to\overline{\R^a(x)}^i$ is a transition in $AD(f)$ for $a=0,1$
      and $\pi(x)\to\pi(\bar{x}^i)$ is a transition in $AD(\tf)$.
    \item\label{lem:tr-in-proj} If $x\to\bar{x}^i$ is a transition in $AD(\tf)$, then there exists $a\in\{0,1\}$
       such that $\S^a(x)\to\overline{\S^a(x)}^i$ is a transition in $AD(f)$.
       In addition, from each $y \in \pi^{-1}(x)$ there exists a path to $\overline{\S^a(x)}^i$.
  \end{enumerate}
\end{lemma}
\begin{proof}
  \begin{enumerate}[label=(\roman*)]
    \item If $\R^0(x)=\R^1(x)$, then $\bar{x}^v=\R^0(x)=\R^1(x)$ and there is a transition
      in $AD(f)$ from $x$ to $\R^0(x)$ and $\R^1(x)$ as a consequence of the definition of asynchronous state transition graph.

      If $\R^0(x)\neq\R^1(x)$, then from~\cref{rmk:R0noteqR1} we have $f_v(x^{v=0})=1$ and $f_v(x^{v=1})=0$,
      and there is a transition in $AD(f)$ from $x^{v=0}$ to $\overline{x^{v=0}}^v=\R^0(x)$ and from $x^{v=1}$ to $\overline{x^{v=1}}^v=\R^1(x)$.

    \item If $\R^0(x)=\R^1(x)$, then $\tf_i(\pi(x))=f_i(\R^0(x))=f_i(\R^1(x))\neq \R^0(x)_i=\R^1(x)_i=x_i$.

      If $\R^0(x)\neq\R^1(x)$, then from~\cref{rmk:R0noteqR1} we have $\R^0(x)=x^{v=1}$, $\R^1(x)=x^{v=0}$,
      and either $f_i(\R^0(x))\neq x_i$ or $f_i(\R^1(x))\neq x_i$.
      If $x_i=1$, then $\tf_i(\pi(x))=f_i(\R^0(x))\wedge f_i(\R^1(x))=0$.
      If $x_i=0$, then $\tf_i(\pi(x))=f_i(\R^0(x))\vee f_i(\R^1(x))=1$, as required.

    \item Since $i$ does not depend on $v$, we have $f_i(\R^0(x))=f_i(\R^1(x))=f_i(x)\neq x_i=\R^0(x)_i=\R^1(x)_i$,
      which gives the first part.
      For the second, it is sufficient to observe that $\tf_i(\pi(x))=f_i(\R^0(x))=f_i(\R^1(x))\neq\pi(x)_i$.

    \item Since $i \neq v$ we have $x_i = \S^0(x)_i=\S^1(x)_i$. Since there is a transition $x\to\bar{x}^i$ in $AD(\tf)$, by definition of $\tf$ either $f_i(\S^0(x)) \neq x_i$ or $f_i(\S^1(x)) \neq x_i$, that is, either $\S^0(x)\to\overline{\S^0(x)}^i$ or $\S^1(x)\to\overline{\S^1(x)}^i$ is a transition in $AD(f)$.

The second part follows from point~\ref{lem:x-to-repr2}.
  \end{enumerate}
\end{proof}

The following result generalizes \cref{thm:prev-results}. For \cref{thm:fixed-points} (iii) note that if $v$ is not autoregulated the set $\{\S^0(x),\S^1(x)\}$ has cardinality one, hence it is a fixed point and the result generalizes \cref{thm:prev-results} (ii).  
\begin{theorem}\label{thm:fixed-points}
  Consider a Boolean network $f\colon\B^n\to\B^n$ such that there is no positive loop at $v$ in $G(f)$.
 Then:
  \begin{enumerate}[label=(\roman*)]
    \item if $x\in\B^n$ is a fixed point for $f$, then $x=\R^0(x)=\R^1(x)$, $\pi(x)$ is a fixed point for $\tf$ and no other
                          fixed point for $f$ is projected on $\pi(x)$.
    \item if $\{x,\bar{x}^v\}$ is a cyclic attractor for $AD(f)$, then $\pi(x)$ is a fixed point for $\tf$.
    \item if $x\in\B^{n-1}$ is a fixed point for $\tf$, then the set $\{\S^0(x),\S^1(x)\}$
      is an attractor of $AD(f)$.
    \item if $T\subseteq\B^n$ is a trap set for $f$, then $\pi(T)$ is a trap set for $\tf$.
    \item if $\{x,\bar{x}^i\}$ is a cyclic attractor of $AD(f)$ for some $x\in\B^n$ and $i\neq v$,
      then the set $\{\pi(x),\pi(\bar{x}^i)\}$ is a cyclic attractor of $AD(\tf)$.
    \item if $\tA\subseteq\B^{n-1}$ is an attractor for $\tf$,
       there exists at most one attractor for $AD(f)$ intersecting $\pi^{-1}(\tA)$.
  \end{enumerate}
\end{theorem}

\begin{proof}
  \begin{enumerate}[label=(\roman*)]
    \item Since $x$ is fixed, by~\cref{lemma2}~\ref{lem:x-to-repr2} we have $x=\R^0(x)=\R^1(x)$.
          $\pi(x)$ is fixed for $\tf$ as a consequence of~\cref{lemma2}~\ref{lem:tr-in-proj},
          and the absence of a positive loop at $v$ gives that $\bar{x}^v$ is not fixed.
    \item Consequence of~\cref{lemma2}~\ref{lem:tr-in-proj}.
    \item The set $\{\S^0(x),\S^1(x)\}$ consists either of one state if $\S^0(x)=\S^1(x)$, or two strongly connected
          states if $\S^0(x)\neq\S^1(x)$ (\cref{lemma2}\ref{lem:x-to-repr2}). In addition, the set is a trap set by~\cref{lemma2}~\ref{lem:repr-to-any2}.
    \item If $x\in\pi(T)$ and $x\to\bar{x}^i$ is a transition in $AD(\tf)$, by~\cref{lemma2}~\ref{lem:tr-in-proj}
          there is a transition $\S^j(x)\to\overline{\S^j(x)}^i$ for some $j\in\{0,1\}$.
          Take $y\in T$ such that $\pi(y)=x$.

          If $\S^0(x)=\S^1(x)$, then either $y=\S^0(x)=\S^1(x)$ or, by~\cref{lemma2}\ref{lem:x-to-repr2}, there is a transition from $y$ to $\S^0(x)=\S^1(x)$.
          If $\S^0(x)\neq \S^1(x)$, then $\S^0(x)$ and $\S^1(x)$ are strongly connected by~\cref{lemma2}\ref{lem:x-to-repr2} and hence belong to $T$.
          In both cases $\overline{\S^j(x)}^i$ is in $T$ and $\bar{x}^i=\pi(\overline{\S^j(x)}^i) \in \pi(T)$.
    \item Since $\{x,\bar{x}^i\}$ is a cyclic attractor of $AD(f)$ we have $f_v(x)=f_v(\bar{x}^i)=x_v$ and therefore $\R^0(x)=\R^1(x)=x$, $\R^0(\bar{x}^i)=\R^1(\bar{x}^i)=\bar{x}^i$.
          Then $\{\pi(x),\pi(\bar{x}^i)\}$ is strongly connected as a consequence of~\cref{lemma2}~\ref{lem:repr-to-any2}.
          It is a trap set by the previous point.
    \item From point $(i)$ of~\cref{lemma2}, we know that from all states in $\B^n$ there is a transition to $\R^0(\B^n)\cup\R^1(\B^n)=\S^0(\pi(\B^n))\cup\S^1(\pi(\B^n))$. Hence it is sufficent to show that, for each pair $x,y\in\S^0(\tA)\cup\S^1(\tA)$, there exists a path from $x$ to $y$ in $\pi^{-1}(\tA)$.

          Write $x=\S^j(a)$, $y=\S^k(b)$ for some $j,k\in\{0,1\}$ and $a,b\in\tA$.
          Since $\tA$ is strongly connected, there exists a path from $a$ to $b$ in $\tA$.
          Consider a transition $c\to\bar{c}^i$ in this path.
          By~\cref{lemma2}\ref{lem:tr-in-proj}, there exists $h\in\{0,1\}$ such that there exist paths from $\S^0(c)$ and from $\S^1(c)$ to $\overline{\S^h(c)}^i$ in $AD(f)$, with $\pi(\overline{\S^h(c)}^i)=\bar{c}^i$.
          By point $(i)$ of~\cref{lemma2} there exists a path from $\overline{\S^h(c)}^i$ to $\R^0(\overline{\S^h(c)}^i)=\S^0(\pi(\overline{\S^h(c)}^i))=\S^0(\bar{c}^i)$ and $\R^1(\overline{\S^h(c)}^i)=\S^1(\pi(\overline{\S^h(c)}^i))=\S^1(\bar{c}^i)$, which concludes.
  \end{enumerate}
\end{proof}

Denote by $S(f)$ the number of fixed points of $f$, by $A(f)$ the number of cyclic attractors of $f$,
and by $A(f,i)$ the number of cyclic attractors of $AD(f)$ consisting of two states that differ in component $i$.
\begin{corollary}\label{cor:fixed-points-and-2}
  If $\tf$ is obtained from $f$ by eliminating component $v$, then
  \begin{enumerate}[label=(\roman*)]
    \item\label{cor:equality-bound-fixed} $S(\tf) = S(f) + A(f, v)$ and hence $S(f)\leq S(\tf)$.
    \item\label{cor:bound2} For all $i\neq v$, $A(f, i)\leq A(\tf, i)$.
    \item\label{cor:bound} $S(f)+A(f)\leq S(\tf)+A(\tf)$.
  \end{enumerate}
\end{corollary}
\begin{proof}
    $(i)$ is a corollary of points $(i)$, $(ii)$ and $(iii)$ of~\cref{thm:fixed-points}.
    $(ii)$ is a consequence of point $(v)$ of~\cref{thm:fixed-points},
    and $(iii)$ of part $(vi)$ of~\cref{thm:fixed-points}.
\end{proof}

The inequalities of the corollary can be strict.
For point $(i)$, take $f(x_1,x_2)=(1,\bar{x}_1\vee\bar{x}_2)$. Then $S(f)=0$, $A(f)=A(f,2)=1$
and after removing the second component we have $\tf(x_1)=1$, $S(\tf)=1$, $A(\tf)=A(\tf,1)=0$.
For point $(ii)$, the map $f(x_1,x_2)=(\bar{x}_2,x_1)$ after removing variable $x_2$
gives $\tf(x_1)=\bar{x}_1$, $A(f)=A(f,1)=A(f,2)=0$, $A(\tf)=A(\tf,1)=1$.
For the third point, see \cref{ex:forward}.

\subsection{Interaction graph}

The following result is a consequence of the properties of the reduction method described in~\cite{naldi2009reduction}.
It states that the classical reduction cannot introduce new paths in the interaction graph:
if a path exists in the interaction graph of the reduced Boolean network, a path of the same sign must exist in the interaction graph of the original network.
We will prove here a generalized version for the case of the removal of potentially negatively autoregulated components.

\begin{proposition}\label{prop:ig-classical}
  If $G(f)$ has no loops at $v$, and $G(\tf)$ has a path from $j$ to $i$ of sign $s$,
  then $G(f)$ has a path from $j$ to $i$ of sign $s$.
\end{proposition}

The goal of this section is to generalize this result to negatively autoregulated components that are removed. However, as the following example shows we need to be careful here.  Indeed, if $G(f)$ has a negative loop at $v$,
then the conclusion of~\cref{prop:ig-classical} does not necessarily hold for paths that are negative loops.

\begin{example}\label{ex:neg-loop-case}
  The Boolean network $f(x_1,x_2)=(\bar{x}_2,\bar{x}_2)$ reduces to $\tf(x_1)=\bar{x}_1$ when the
  second variable is eliminated. The graph $G(\tf)$ has a negative loop in $1$,
  whereas $G(f)$ has no negative circuit containing $1$.
\end{example}

We first examine the negative loop case, then show that the result in~\cref{prop:ig-classical} can be extended
to the reduction method intruduced in this paper for the case of positive loops
and paths of length at least two.

\begin{proposition}\label{prop:ig-neg-loop}
  If $G(\tf)$ has a negative loop at $i$, then $G(f)$ has a negative cycle with support contained in $\{i,v\}$.
\end{proposition}
\begin{proof}
  Take $x$ such that $G(\tf)$ has a negative loop at $x$ and $x_i=0$, so that $\tf_i(x)=1$ and $\tf_i(\bar{x}^i)=0$.
  Consider $y\in\B^n$ such that $\pi(y)=x$, then, by definition of $\tf$, either $f_i(\R^0(y))$ or $f_i(\R^1(y))$ is equal to $1$,
  and either $f_i(\R^0(\bar{y}^i))$ or $f_i(\R^1(\bar{y}^i))$ is equal to $0$.
  Consider two cases.

  $(1)$ Suppose that there exists $a\in\{0,1\}$ such that $f_i(\R^a(y))=1$ and $f_i(\R^a(\bar{y}^i))=0$.
  Then
  \begin{equation*}
    \begin{aligned}
      (-1)\cdot(\bar{y}^i_i-y_i)&=\tf_i(\bar{x}^i)-\tf_i(x)=f_i(\R^a(\bar{y}^i))-f_i(\R^a(y)) \\
                               &=f_i(\R^a(\bar{y}^i))-f_i(\overline{\R^a(y)}^i)+f_i(\overline{\R^a(y)}^i)-f_i(\R^a(y)).
    \end{aligned}
  \end{equation*}
  If $f_i(\overline{\R^a(y)}^i)=0$, then there is a negative loop with support $\{i\}$ at $\R^a(y)$.
  Otherwise, we have $\R^a(\bar{y}^i)\neq\overline{\R^a(y)}^i$ and $f_v(\overline{y^{v=a}}^i)\neq f_v(y^{v=a})$.
  Hence
  \begin{equation*}
    -1 = \frac{f_i(\R^a(\bar{y}^i))-f_i(\R^a(y))}{\bar{y}^i_i-y_i} =
         \frac{f_i(\R^a(\bar{y}^i))-f_i(\overline{\R^a(y)}^i)}{f_v(\overline{y^{v=a}}^i)-f_v(y^{v=a})}\frac{f_v(\overline{y^{v=a}}^i)-f_v(y^{v=a})}{\bar{y}^i_i-y_i},
  \end{equation*}
  that is, there is a negative cycle in $G(f)$ with support $\{i,v\}$.

  $(2)$ Suppose now that, for $a=0$ and $a=1$, if $f_i(\R^a(y))=1$ then $f_i(\R^a(\bar{y}^i))=1$,
  and if $f_i(\R^a(\bar{y}^i))=0$ then $f_i(\R^a(y))=0$.
  Then we must have $f_i(\R^0(y))\neq f_i(\R^1(y))$ and $f_i(\R^0(\bar{y}^i))\neq f_i(\R^1(\bar{y}^i))$.
  In particular, $\R^0(y)\neq\R^1(y)$ and by~\cref{rmk:R0noteqR1} there is a negative loop at $v$ in $G(f)$.
\end{proof}

\begin{proposition}\label{prop:ig-pos-loop}
  If $G(\tf)$ has a positive loop at $i$, then $G(f)$ has either a positive loop at $i$ or a positive cycle with support $\{i,v\}$.
\end{proposition}
\begin{proof}
  Take $x$ such that $G(\tf)(x)$ has a positive loop at $i$ and w.l.o.g. $x_i=0$, so that $\tf_i(x)=0$ and $\tf_i(\bar{x}^i)=1$.
  Consider $y\in\B^n$ such that $\pi(y)=x$, then, by definition of $\tf$, $f_i(\R^0(y))=f_i(\R^1(y))=0$ and $f_i(\R^0(\bar{y}^i))=f_i(\R^1(\bar{y}^i))=1$.
  We can write
  \begin{equation*}
    \bar{y}^i_i-y_i=f_i(\R^0(\bar{y}^i))-f_i(\R^0(y))=f_i(\R^0(\bar{y}^i))-f_i(\overline{\R^0(y)}^i)+f_i(\overline{\R^0(y)}^i)-f_i(\R^0(y)).
  \end{equation*}
  If $f_i(\overline{\R^0(y)}^i)=1$, then the equality can be simplified to $\bar{y}^i_i-y_i=f_i(\overline{\R^0(y)}^i)-f_i(\R^0(y))$ and there is a positive loop at $i$ in $G(f)(\R^0(y))$.
  Otherwise, we have $\R^0(\bar{y}^i)\neq\overline{\R^0(y)}^i$ and $f_v(\overline{y^{v=0}}^i)\neq f_v(y^{v=0})$.
  Hence
  \begin{equation*}
    1 = \frac{f_i(\R^0(\bar{y}^i))-f_i(\R^0(y))}{\bar{y}^i_i-y_i} =
        \frac{f_i(\R^0(\bar{y}^i))-f_i(\overline{\R^0(y)}^i)}{f_v(\overline{y^{v=0}}^i)-f_v(y^{v=0})}\frac{f_v(\overline{y^{v=0}}^i)-f_v(y^{v=0})}{\bar{y}^i_i-y_i},
  \end{equation*}
  that is, there is a positive cycle in $G(f)$ with support $\{i,v\}$.
\end{proof}

\begin{proposition}\label{prop:edges-in-ig}
  If $G(\tf)$ has an edge from $j$ to $i$ of positive (resp. negative) sign and $i\neq j$,
  then $G(f)$ has an edge $j\to i$ or a path $j\to v\to i$ of positive (resp. negative) sign.
\end{proposition}
\begin{proof}
  Suppose that $\tf_i(\bar{x}^j)\neq\tf_i(x)$.
  Take $y\in\B^n$ such that $\pi(y)=x$. From~\cref{eq:casesftilde} we have
  \begin{equation*}
  \tf_i(x)=\begin{cases}
             f_i(\R^0(y))\wedge f_i(\R^1(y)) & \text{ if }x_i=1,\\
             f_i(\R^0(y))\vee f_i(\R^1(y)) & \text{ if }x_i=0,
           \end{cases}
  \end{equation*}
  and
  \begin{equation*}
  \tf_i(\bar{x}^j)=\begin{cases}
             f_i(\R^0(\bar{y}^j))\wedge f_i(\R^1(\bar{y}^j)) & \text{ if }\bar{x}^j_i=x_i=1,\\
             f_i(\R^0(\bar{y}^j))\vee f_i(\R^1(\bar{y}^j)) & \text{ if }\bar{x}^j_i=x_i=0.
           \end{cases}
  \end{equation*}
  Since $i\neq j$ and $\tf_i(\bar{x}^j)\neq\tf_i(x)$, we must have either $f_i(\R^0(y))=f_i(\R^1(y))$ or $f_i(\R^0(\bar{y}^j))=f_i(\R^1(\bar{y}^j))$.
  Consider the case where $f_i(\R^0(\bar{y}^j))=f_i(\R^1(\bar{y}^j))$,
  the case $f_i(\R^0(y))=f_i(\R^1(y))$ being symmetrical. Then $\tf_i(\bar{x}^j)=f_i(\R^0(\bar{y}^j))=f_i(\R^1(\bar{y}^j))$.

  Take $a\in\{0,1\}$ such that $\tf_i(x)=f_i(\R^a(y))$.
  Then we can write
    \begin{equation*}
      \begin{aligned}
        0\neq s \cdot (\bar{y}^j_j-y_j) & = f_i(\R^a(\bar{y}^j))-f_i(\R^a(y)) \\
                                        & = f_i(\R^a(\bar{y}^j))-f_i(\overline{\R^a(y)}^j)+f_i(\overline{\R^a(y)}^j)-f_i(\R^a(y)),
      \end{aligned}
    \end{equation*}
    where $s$ is the sign of the edge $j\to i$ in $G(\tf)$.
    If $f_i(\overline{\R^a(y)}^j)-f_i(\R^a(y))=s \cdot (\bar{y}^j_j-y_j)$, then there is an edge $j\to i$ in $G(f)$ with the required sign.

    Suppose that $f_i(\overline{\R^a(y)}^j)-f_i(\R^a(y))=0$, then $f_i(\R^a(\bar{y}^j))-f_i(\overline{\R^a(y)}^j)=s \cdot (\bar{y}^j_j-y_j)$
    and $f_v(\overline{y^{v=a}}^j)\neq f_v(y^{v=a})$.
    Therefore
    \begin{equation*}
      \begin{aligned}
        \frac{f_i(\R^a(\bar{y}^j))-f_i(\overline{\R^a(y)}^j)}{f_v(\overline{y^{v=a}}^j)-f_v(y^{v=a})}\cdot
        \frac{f_v(\overline{y^{v=a}}^j)-f_v(y^{v=a})}{\overline{y^{v=a}}^j_j-y^{v=a}_j}=
        \frac{f_i(\R^a(\bar{y}^j))-f_i(\R^a(y))}{\bar{y}^j_j-y_j}=s,
      \end{aligned}
    \end{equation*}
    and there is a path $j\to v\to i$ in $G(f)$ with the required sign.
\end{proof}

The following is a corollary of the previous proposition.

\begin{proposition}\label{prop:ig}
  If $G(\tf)$ has an elementary path from $j$ to $i$ of sign $s$ that is not a loop,
  then $G(f)$ has an elementary path from $j$ to $i$ of sign $s$.
\end{proposition}

\begin{example}
  Not all edges in $G(f)$ are preserved by the reduction.
  For instance, the map $f(x_1,x_2)=(x_1\wedge \bar{x}_2,x_1\wedge \bar{x}_2)$ reduces,
  after elimination of the second variable, to the constant function $\tilde{f}(x_1)=0$.
\end{example}

\subsection{Application: positive feedback vertex sets and bound on the number of attractors}

Using the properties of the variable elimination method introduced in this paper,
we give an alternative proof for a bound on the number of attractors of asynchronous Boolean networks in terms of positive feedback vertex sets of the interaction graph. The result can be found in~\cite{richard2009positive}.
Recall that a \emph{positive feedback vertex set} of a signed directed graph $G$
is a set of vertices that intersects every positive cycle of $G$.

The idea is to show that the size of the minimum positive feedback vertex set does not increase with the reduction. After some reduction steps a network with $|I|$ components is obtained, giving the upper bound $2^{|I|}$ on the number of attractors and fixed points. For the proof we use the following lemma.
\begin{lemma}\label{lem:pfvs}
  Suppose that $I$ is a positive feedback vertex set for $f$ that does not contain $v$.
  Then $G(f)$ has no positive loop at $v$,
  and $I$ is a positive feedback vertex set for the network $\tf$ obtained by eliminating $v$.
\end{lemma}
\begin{proof}
  Take a positive cycle in $G(\tf)$ with support $C$. By~\cref{prop:ig-pos-loop,prop:ig} there exists
  a positive cycle in $G(f)$ with support in $C\cup\{v\}$.
  Since $I$ is a positive feedback vertex set and $v$ is not in $I$, we have $I\cap C\neq\emptyset$, as required.
\end{proof}

\begin{theorem}\label{thm:bound-attrs}(\cite{richard2009positive})
  Consider $f\colon\B^n\to\B^n$ and suppose that $I$ is a positive feedback vertex set of $G(f)$.
  Then $AD(f)$ has at most $2^{|I|}$ attractors.
\end{theorem}
\begin{proof}
  We can apply the reduction method described in~\cref{eq:casesftilde}
  eliminating vertices that do not belong to a positive feedback vertex set of minimum size,
  until a network $\tf\colon\B^m\to\B^m$ is obtained such that all positive
  feedback vertex sets have size $m$.
  Since variables that do not belong to minimum positive vertex sets are eliminated,
  by~\cref{lem:pfvs} the size of minimum positive feedback vertex sets cannot increase with the reduction.
  The conclusion follows from the fact that the number of attractors of $\tf$ is greater or equal to the number of attractors of $f$ (\cref{cor:fixed-points-and-2}~\ref{cor:bound}).
\end{proof}

\section{Preservation of cyclic attractors}\label{sec:attractor-preservation}
We now turn our attention to a different problem. A critical question when using network reduction concerns the preservation or loss of information.
The identification of properties that are preserved can help clarify the accuracy of information that can be obtained from the analysis of the reduced network in lieu of the full network.
When studying a network model one should also consider that, even if no network reduction is explicitly applied, implicit reduction steps might have been introduced in the construction of the model, for instance when certain components are merged into one, or signaling pathways are simplified.
In the analysis of Boolean networks special importance is given to the attractors. A natural question is therefore: under which conditions are attractors \emph{preserved} by the network reduction?

\subsection{Definition and examples}

To express and illustrate structural conditions on the interaction graphs, in this section we will adhere to the following conventions.
We will write $i \xrightarrow{s} j$ for an edge with sign $s$ from $i$ to $j$, whereas $i\to j$ will denote the existence of an edge of any sign.
In addition, to represent classes of interaction graphs in compact form, we will summarize subgraphs using subsets of vertices.
For instance, given $X,Y \subseteq V$, $X\to Y$ will denote an interaction graph consisting of arbitrary signed directed graphs with vertices in $X$ and $Y$ respectively, and at least one edge from some variable in $X$ to some variable in $Y$.
We will also denote the possibility of existence of an edge from a vertex to another using dashed arrows.
Thus, for instance, $X\dashedrightarrow Y$ will denote the possible existence of an edge from some variable in $X$ to some variable in $Y$.

Before we answer the question posed in the introduction of this section, we need to clarify the meaning of the term ``preservation''.
In agreement with Definition~$2.3$ in \cite{saadatpour2013reduction}, we consider the following definition.

\begin{definition}\label{def:strong-pres-attr}
  We say that the attractors of $f$ are \emph{preserved} by the elimination of $v$ if the following two conditions are satisfied:
  \begin{itemize}
    \item[(i)] for each attractor $A$ of $AD(f)$, $\pi(A)$ is an attractor of $AD(\tf)$, and
    \item[(ii)] for each attractor $\tA$ of $AD(\tf)$, there exists a unique attractor $A$ of $AD(f)$ such that $\pi(A)=\tA$.
  \end{itemize}
\end{definition}
Note that, if attractors are preserved, their number cannot change as a result of the reduction.

\begin{example}\label{ex:forward}
  Consider the map defined by
  \begin{equation*}
    f(x_u,x_v,x_w)=(\bar{x}_u,x_u,(x_u\wedge x_w)\vee(\bar{x}_v\wedge x_w)\vee(x_u\wedge \bar{x}_v))
  \end{equation*}
  with interaction graph
  
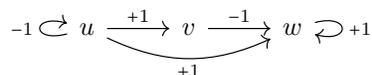
\begin{equation*}
    \begin{tikzcd}
      u \arrow[r,"+1"] \arrow[rr,bend right=25,"+1"'] \arrow[loop left,->,"-1"] & v \arrow[r,->,"-1"] & w \arrow[loop right,"+1"]
    \end{tikzcd}
  \end{equation*}
  Then $f$ has only one attractor (the full space), whereas the map obtained by elimination of $v$ has two attractors.
  The asynchronous state transition graphs for $f$ and $\tf$ are represented in~\cref{fig:ex-forward}.

  \begin{figure}
  \centering
  \begin{tikzcd}[column sep=tiny,row sep=tiny]
    & 011 \arrow[dl] \arrow[dd] \arrow[rr,yshift=+1pt] & & {\color{black}111} \arrow[ll,yshift=-1pt] \\
    {\color{black}001} \arrow[rr,yshift=+1.5pt,crossing over] & & 101 \arrow[ll,yshift=-1.0pt] \arrow[ru] & \\
    & 010 \arrow[dl] \arrow[rr,yshift=-1pt] & & {\color{black}110} \arrow[ll,yshift=+1pt] \\
    {\color{black}000} \arrow[rr,yshift=-1pt] & & 100 \arrow[ll,yshift=+1pt] \arrow[uu,crossing over] \arrow[ur] &
  \end{tikzcd}\qquad
  \begin{tikzcd}
    01 \arrow[r,yshift=+1pt] & 11 \arrow[l,yshift=-1pt] \\
    00 \arrow[r,yshift=+1pt] & 10 \arrow[l,yshift=-1pt]
  \end{tikzcd}\caption{Asynchronous state transition graphs for the map $f(x_u,x_v,x_w)=(\bar{x}_u,x_u,(x_u\wedge x_w)\vee(\bar{x}_v\wedge x_w)\vee(x_u\wedge \bar{x}_v))$
  (left) and the one obtained from $f$ by eliminating $v$ (right). The state transition graphs have one cyclic attractor and two cyclic attractors respectively.}\label{fig:ex-forward}
  \end{figure}
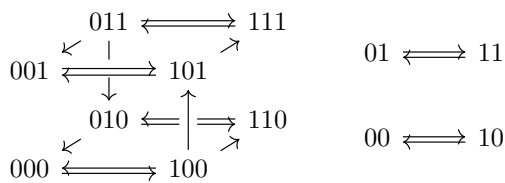

\end{example}

The example shows how the removal of a simple mediator vertex (a component with only one regulator and only one target) can have an impact on the number of attractors.
In \cite{saadatpour2013reduction}, the authors claim that the removal of a mediator variable $v$ does not impact the number of attractors if the regulator and target of $v$ are not regulators of each other.
Here we show that, on the contrary, the removal of a single mediator vertex from a chain of mediator variables of arbitrary length can change the number of attractors.
In other words, for each $n\geq 1$, one can construct a map with interaction graph of the form
\begin{equation}\label{eq:ig-counterex}
  \begin{tikzcd}
    U \arrow[r] & v_1 \arrow[r] & v_2 \arrow[r] & \cdots \arrow[r] & v_{n+1} \arrow[r] & W \arrow[lllll,bend left=20]
  \end{tikzcd}
\end{equation}
such that the reduced network obtained by eliminating any of the variables $v_i$, $i=1,\dots,n+1$, has a different number of cyclic attractors.

Write $\0$ and $\1$ for the states with all components equal to 0 or 1, respectively.
The idea for the construction of the map is as follows.
We have a chain of $n+1$ mediator variables $v_1,\dots,v_{n+1}$.
Downstream of the chain, the networks has $n+2$ variables $W$, that regulate each other and have $v_{n+1}$, the last mediator variable, as their unique regulator outside of $W$.
The variables in $W$ regulate a variable $u$ whose unique role is to regulate the first mediator variable $v_1$.
We want the initial network to have a unique steady state in $\1$. We build it so that,
starting from the state $\0$, exactly $n+1$ mediator variables are required to switch on all the variables in $W$, so that
the state $\1$ cannot be reached from $\0$ when one of the mediator variables is removed, and another attractor must exist in the reduced network.

To impose such behaviour, we build the following dependencies for the variables in $W$: they can only be updated from $0$ to $1$ in order (first $w_1$, then $w_2$, etc.), and the variables in odd positions require, in order to change from $0$ to $1$, the condition $v_{n+1}=1$ for the last mediator variable, whereas the variables in even positions require $v_{n+1}=0$.
Therefore, in order to reach the state $\1$, we are forced to have alternating values for the mediator variables.

The values of the mediator variables are simply propagated from the variable $u$.
The Boolean function $f_u$ is defined so that $x_u$ is forced to be $1$ when some variables in $W$ are equal to $1$, and can otherwise oscillate freely, thus providing the oscillating input to the chain of mediator nodes.

The detailed definition of the map is given in the following thereom.

\begin{theorem}\label{thm:chain}
  For each $n\geq 1$ there exists a Boolean network $f$ with interaction graph $G(f)$ that admits a path of length $n$
  of variables with indegree one and outdegree one, such that the network $\tf$ obtained from $f$ by removing any of
  the variables in the path satisfies $S(\tf)>S(f)$.
\end{theorem}
\begin{proof}
  We define a Boolean network $f$ of dimension $2(n+2)$ with interaction graph of the form
  given in~\cref{eq:ig-counterex}.
  Denote $u,v_1,\dots,v_{n+1},w_1,\dots,w_{n+2}$ the variables of the Boolean network.

  Set $V=\{v_1,\dots,v_{n+1}\}$, $W=\{w_1,\dots,w_{n+2}\}$, and
  \begin{equation*}
    \begin{aligned}
      f_u(x_u,x_V,x_W)&=\left(\bigwedge_{j=1}^{n+2}x_{w_j}\right)\vee
                       \left(\bar{x}_u\wedge\bigwedge_{j=1}^{n+2}\bar{x}_{w_j}\right)\vee
                       \left(x_u\wedge\overline{\bigwedge_{j=1}^{n+2}x_{w_j}}\wedge\overline{\bigwedge_{j=1}^{n+2}\bar{x}_{w_j}}\right)\\
      f_{v_1}(x_u,x_V,x_W)&=x_{u},\\
      f_{v_i}(x_u,x_V,x_W)&=x_{v_{i-1}} \ \text{for }i=2,\dots,n+1,\\
      f_{w_i}(x_u,x_V,x_W)&=\left(\bigwedge_{j=1}^{n+2}x_{w_j}\right)\vee
                          \left(x_{v_{n+1}}\wedge\bigwedge_{j=1}^{i-1}x_{w_j}\wedge\bigwedge_{j=i}^{n+2}\bar{x}_{w_j}\right)
                          \ \text{for }i=1,\dots,n+2\ \text{if } i\equiv 1\ \mathrm{mod}\ 2,\\
      f_{w_i}(x_u,x_V,x_W)&=\left(\bigwedge_{j=1}^{n+2}x_{w_j}\right)\vee
                          \left(\bar{x}_{v_{n+1}}\wedge\bigwedge_{j=1}^{i-1}x_{w_j}\wedge\bigwedge_{j=i}^{n+2}\bar{x}_{w_j}\right)
                          \ \text{for }i=1,\dots,n+2\ \text{if } i\equiv 0\ \mathrm{mod}\ 2.
    \end{aligned}
  \end{equation*}
  The variables $v_1,\dots,v_{n+1}$ are a chain of variables with indegree and outdegree equal to one in the interaction graph of $f$.
  The first conjunction in each function is to ensure that $\1$ is a fixed point.
   The definition of $f_u$ is such that component $u$ can be changed from $0$ to $1$ and vice versa, as long as the variables in $W$ are all equal to $0$.
   The definition of $f_{w_i}$ shows a dependency on the previous variables $w_1,\dots,w_{i-1}$, and is different for variables in odd and even positions in $W$ in terms of the dependency on the last mediator variable $v_{n+1}$.
  
  We prove that $(a)$ for $AD(f)$, the fixed point $\1$ is the unique attractor and $(b)$ $AD(\tf)$ has an additional cyclic attractor.

  $(a)$ It is easy to see that $z=\1$ is a fixed point for $f$.
  We show that, for each $y\in\B^{2n+4}$, $y\neq z$, there exists a path in $AD(f)$ from $y$ to $z$.
  \begin{itemize}
    \item[(i)] Case $y_W=\1$: we have $f_u(y)=1$, and the reachability of $z$ from $y$
      is direct from the definition of $f$.
    \item[(ii)] Case $y_W=\0$: first observe that there is a path from $y$ to $y^1$ that satisfies $y^1_u=1$, $y^1_{v_i}=1$ for $i=1,\dots,n+1$ and $y^1_W=\0$.

    From $y^1$, one can switch component $u$ to zero and construct a path to the state $y^2$ defined by $y^2_u=0$, $y^2_{v_i}=0$ for $i=1,\dots,n$, $y^2_{v_{n+1}}=1$ and $y^2_W=\0$.

    Then component $u$ can be switched back to one, and its value propagated to component $v_{n-1}$, while keeping component $v_n$ to zero.
    Continuing with this construction, one can reach a state $y'$ that satisfies
    $$y'_{v_{n+2-i}}\equiv i\ \mathrm{mod}\ 2\ \text{for }i=1,\dots,n+1, y'_W=0.$$

    Using $y'_{v_{n+1}}=1$ one can then update the value of component $w_1$.
    After this, to update the value of $w_2$, one needs to first propagate the value of component $v_{n}$ to change component $v_{n+1}$ to zero, and so forth.
    One can therefore reach states $z^1,\dots,z^{n+1}$ that satisfy
    \begin{align*}
      & z^1_{v_{n+1}}=1,z^1_{w_1}=1,z^1_{w_2}=0,\dots,z^1_{w_{n+1}}=0,\\
      & z^2_{v_{n+1}}=0,z^2_{w_1}=1,z^2_{w_2}=1,z^2_{w_3}=0,\dots,z^2_{w_{n+1}}=0,\\
      & \cdots\\
      & z^{n+1}_{v_{n+1}}\equiv n+1\ \mathrm{mod}\ 2,z^{n+1}_W=\1.
    \end{align*}
    Hence, we have a path from $y$ to $z$ by point $(i)$.

    \item[(iii)]
      In the remaining cases, it is easy to see that there is a path from $y$ to a state $y'$ with $y'_W=\0$.
      We conclude using point $(ii)$.
  \end{itemize}

  $(b)$ Consider now the asynchronous state transition graph for the network $\tf\colon\B^{2n+3}\to\B^{2n+3}$ obtained from $f$
  by eliminating one of the variables $v_1,\dots,v_n$.
  Without loss of generality, we can consider the case where $v_1$ is eliminated.
  Consider the set of states $A$ reachable from $\0\in\B^{2n+3}$ in $AD(\tf)$, and define $\alpha = \max_{y\in A} \sum_{j=1}^{n+1}y_{w_j}$.
  We show that $\alpha\leq n$, and therefore $\1\notin A$, and $AD(\tf)$ admits
  at least one attractor distinct from $\1$.

  Take any $y \in A$ and a path from $\0$ to $y$, and call $y'$ the last state in the path that satisfies $y'_W=\0$.

  Denote by $p$ the path from $y'$ to $y$. Note that component $u$ does not change in $p$.
  In addition, $\alpha-1$ is bounded by the number of times component $v_{n+1}$ changes in $p$.
  Since component $u$ is fixed in $p$, $\alpha-1$ is bounded by the cardinality
  of $\{i\in\{2,\dots,n\}\ |\ y'_{v_i}\neq y'_{v_{i+1}}\}$.
  Hence $\alpha$ is bounded by $n$, which concludes.
\end{proof}

The result demonstrates how a very modest change in the interaction graph can have a significant impact on the asymptotic behaviour of asynchronous dynamics.
In the next section we show that attractors are preserved if cycles containing an intermediate variable are not allowed, and the regulators of the intermediate variables do not directly regulate the targets of the intermediate variables.

\subsection{Sufficient conditions}

We now give sufficient conditions on the interaction graph of a Boolean network
for the preservation of attractors to hold.
For the proof we will use the following lemma.

\begin{lemma}\label{lemma:paths-I}
  Consider a Boolean network $f$ with set of components $V$.
  Take $W \subset V$ and $I \subseteq V \setminus W$, $I\neq\emptyset$ and
  suppose that for all $j\in W$ there is no path from $j$ to $I$ in $G(f)$.
  If there exists a path in $AD(f)$ from a state $x$ to a state $y$ with $y_I=\bar{x}_I$,
  then there exists a path in $AD(f)$ from $x$ to a state $z$ such that $z_I=\bar{x}_I$ and $z_W=x_W$.
\end{lemma}
\begin{proof}
  Write $W'$ for the vertices that are reachable from $W$ in $G(f)$.
  Observe that, if $u\to\bar{u}^i$ is a transition in $AD(f)$ and $i$ is not reachable from $W$ in $G(f)$, then, for any subset $W''$ of $W'$, the transition $\bar{u}^{W''}\to\bar{u}^{W''\cup\{i\}}$ exists in $AD(f)$.

  Now consider a path $x = x^{1}\rightarrow \dots \rightarrow x^{m} = y$ from $x$ to $y$, and write $i_1,\dots,i_m$ for the sequence of components being updated along the path from $x$ to $x^m$.
  Consider the subsequence obtained from $i_1,\dots,i_m$ by removing all indices in $W'$.
  Then, by the previous observation, the subsequence defines a trajectory in $AD(f)$ from $x$ to a state $z$ that satisfies $z_I=y_I=\bar{x}_I$ and $z_W=x_W$.
\end{proof}

\begin{theorem}\label{thm:pres-attr}
  Suppose that the interaction graph of $f$ is of the form
  \begin{equation}\label{eq:graph-main}
    \begin{tikzcd}
      U_1 \arrow[r,dashed] \arrow[rrr,dashed,bend right=20] & U_2 \arrow[r] & v \arrow[loop,out=130,in=50,dashed,looseness=4,"-1"] \arrow[r] & W
    \end{tikzcd}
  \end{equation}
  for some $U_1,U_2,W \subset V$, $v \in V$.
  Then the attractors of $f$ are preserved by the elimination of $v$.
\end{theorem}
\begin{proof}
  Write $\tf$ for the network obtained by elimination of $v$, and set $U=U_1\cup U_2$.
  Start by observing that, by~\cref{prop:ig}, the interaction graph $G(\tf)$ takes the form
  \begin{equation*}
    \begin{tikzcd}
      U_1 \arrow[r,dashed] \arrow[rr,dashed,bend right=20] & U_2 \arrow[r] & W.
    \end{tikzcd}
  \end{equation*}
  Without loss of generality, we can write a state $x$ in $\B^n$ as $x=(x_U,x_v,x_W)=(x_{U_1},x_{U_2},x_v,x_W)$.
  In the proof, we use the notation $x \rightsquigarrow y$ to indicate the existence of a path from $x$ to $y$.
  \begin{enumerate}
    \item
    Consider point $(i)$ of~\cref{def:strong-pres-attr}. If $A$ is an attractor for $AD(f)$,
    by~\cref{thm:fixed-points} $(iv)$ $\pi(A)$ is a trap set for $AD(\tf)$. It remains to show that $\pi(A)$ is strongly connected.
    It is sufficient to show that for each transition $x\to\bar{x}^i$ in $AD(f)$ with $x\in A$
    there is a path from $\pi(x)$ to $\pi(\bar{x}^i)$ in $AD(\tf)$.
    By~\cref{lemma2} we only have to consider the case of $i\in W$
    such that $v\to i$ is an edge in $G(f)$, and $x\neq\R^0(x)=\R^1(x)$, that is, $f_i(x)\neq x_i$ and $x_v\neq f_v(x)$.
    In this case, we are not directly guaranteed a transition from $\pi(x)$ to $\pi(\bar{x}^i)$ in $AD(\tf)$, and we have to construct an alternative path.
    The idea is to create a path to a state where component $i$ changes and that is a ``representative state'', so that the transition involving variable $i$ is preserved with the elimination of $v$.

    Since $x_v\in f_v(A)$, there exists a path in $AD(f)$ from $x$ to a state $z\in A$ such that $z_v\neq f_v(z)=x_v\neq f_v(x)$:
    \begin{equation*}
        x = (x_{U_1},x_{U_2},x_v,x_W) \rightsquigarrow z = (z_{U_1},z_{U_2}, z_v, z_W), \ \ x_v = f_v(z) \neq z_v.
    \end{equation*}
    We now apply~\cref{lemma:paths-I} twice:
    \begin{itemize}
      \item component $v$ depends only on $U_2\cup\{v\}$, and there is no path from $W$ to $U_2\cup\{v\}$ in $G(f)$,
            hence we can assume that $z_{W}=x_{W}$, that is
            \begin{equation*}
                x = (x_{U_1},x_{U_2},x_v,x_W) \rightsquigarrow z = (z_{U_1},z_{U_2}, z_v, x_W), \ \ x_v = f_v(z) \neq z_v.
            \end{equation*}
      \item $x$ and $z$ are in the attractor, hence there is a path from $z$ to $x$ in $AD(f)$.
            In particular, there is a path from $z$ to a state $z'$ with $z'_{U_1}=x_{U_1}$:
            \begin{equation*}
                z = (z_{U_1},z_{U_2}, z_v, x_W) \rightsquigarrow z' = (x_{U_1}, z'_{U_2}, z'_v, z'_W).
            \end{equation*}
            Since there is no path from $U_2\cup\{v\}\cup W$ to $U_1$ in $G(f)$, we can assume $z_{U_2\cup\{v\}\cup W}=z'_{U_2\cup\{v\}\cup W}$.
            In particular, we can assume that $z$ satisfies $z_W=x_W$ and $z_{U_1}=x_{U_1}$,
            obtaining a path
            \begin{equation*}
                x = (x_{U_1},x_{U_2},x_v,x_W) \rightsquigarrow z = (x_{U_1},z_{U_2},z_v,x_W) \to y = (x_{U_1},z_{U_2},\bar{z}^v_v,x_W),
            \end{equation*}
            where we defined $y=\bar{z}^v$ and the transition $z \to y$ derives from the definition of $z$.
    \end{itemize}

    We now look at deriving a path in $AD(\tf)$ from this path.
    Since there is no path from $v$ to $U$ in $G(f)$, by~\cref{lemma2}~\ref{lem:not-infl2} there is a path from $\pi(x)=(x_{U_1},x_{U_2},x_W)$ to $\pi(z)=(x_{U_1},z_{U_2},x_W)$ in $AD(\tf)$.
    In addition, since $i$ depends only on $U_1$, $v$ and $W$, we have
    $f_i(y)=f_i(x_{U_1},z_{U_2},x_v,x_W)=f_i(x)\neq x_i=y_i$, and there is a transition from $y$ to $\bar{y}^i$.
    It is easy to see that \cref{lemma2} applies and there is a transition from $\pi(z)=\pi(y)$ to $\pi(\bar{y}^i)$ in $AD(\tf)$.
    In summary, we obtained a path
    \begin{equation*}
        \pi(x)=(x_{U_1},x_{U_2},x_W) \rightsquigarrow \pi(z)=\pi(y)=(x_{U_1},z_{U_2},x_W) \to \pi(\bar{y}^i)=(x_{U_1},z_{U_2},\bar{x}^i_W).
    \end{equation*}

    Since there is a path from $\bar{y}^i$ to $\bar{x}^i$ in $A$, and the variables in $U$ do not depend on $v$,
    we can again apply~\cref{lemma2}~\ref{lem:not-infl2} and find that
    there is a path from $\pi(\bar{y}^{i})$ to $\pi(\bar{x}^i)$ in $AD(\tf)$,
    obtaining a path from $\pi(x)$ to $\pi(\bar{x}^i)$ as needed.

  \item
    Consider now point $(ii)$ of~\cref{def:strong-pres-attr}.
    Given an attractor $\tA$ for $AD(\tf)$,
    by~\cref{thm:fixed-points} $(vi)$ there is at most one attractor intersecting $\pi^{-1}(\tA)$.
    It remains to show that $\pi^{-1}(\tA)$ contains a trap set for $f$.
    To this end, we show that $B = \{x \in \pi^{-1}(\tA)\ | \ x_v \in f_v(\pi^{-1}(\tA))\}$ is a trap set.

    Take $x\in B$ and $\bar{x}^i$ successor for $x$ in $AD(f)$. We have to show that $\bar{x}^i$ is in $B$.

    If $i=v$, then $f_v(x)\neq x_v$ and since $x$ is in $\pi^{-1}(\tA)$, $f_v(x)$ is in $f_v(\pi^{-1}(\tA))$
    and the successor is in $B$.

    For $i\neq v$, since $\pi(x)$ is in $\tA$, it is sufficient to show that there is a path in $AD(\tf)$ from $\pi(x)$ to $\pi(\bar{x}^i)$,
    since this implies that $\pi(\bar{x}^i)$ is in $\tA$ and therefore $\bar{x}^i$ is in $B$.

    As for the first part of the proof,
    by~\cref{lemma2}, we only have to consider the case where there is an edge $v\to i$ in $G(f)$
    and $x\neq\R^0(x)=\R^1(x)$, that is, $f_i(x)\neq x_i$ and $f_v(x)\neq x_v$.

    By definition of $B$, there exists $(z_U,z_W)\in\tA$ and $z_v\in\{0,1\}$ such that $f_v(z)=x_v$ with $z=(z_U,z_v,z_W)$.
    If $z_v=x_v$, then since there is no positive loop at $v$ in $G(f)$ we must have $f_v(\bar{z}^v)=z_v$ and $z$ is a successor for $\bar{z}^v$.
    If instead $z_v\neq x_v$, then $\bar{z}^v$ is a successor for $z$ with $\bar{z}^v_v=x_v$
    In the first case take $y=z$, in the second define $y=\bar{z}^v$. In both cases we have $y_v=x_v$, $y_U=z_U$, $y_W=z_W$.

    By hypothesis, there exists a path in $AD(\tf)$ from $\pi(x)=(x_U,x_W)$ to $\pi(y)=(y_U,y_W)$.
    We apply again~\cref{lemma:paths-I} twice:
    \begin{itemize}
      \item since there is no path from $W$ to $U$ in $G(\tf)$, we can assume that $y_{W}=x_{W}$;
      \item since there is a path from $\pi(y)$ to $\pi(x)$ and no path from $U_2\cup W$ to $U_1$ in $G(\tf)$, we can assume $y_{U_1}=x_{U_1}$.
    \end{itemize}
    We therefore obtained a path
    \begin{equation*}
        \pi(x)=(x_{U_1},x_{U_2},x_W) \rightsquigarrow \pi(y)=(x_{U_1},y_{U_2},x_W).
    \end{equation*}
    Since $i$ does not depend on $U$, we have $f_i(y)=f_i(x_{U_1},y_{U_2},x_v,x_W)=f_i(x)\neq x_i=y_i$ and again it is easy to see that point~\ref{lem:repr-to-any2} of~\cref{lemma2} applies and there is a transition from $\pi(y)=(y_U,x_W)$ to $\pi(\bar{y}^i)=(y_U,\bar{x}_W^{i})$.

    Finally, there is a path in $AD(\tf)$ from $\pi(\bar{y}^i)=(y_U,\bar{x}_W^{i})$ to $\pi(x)=(x_U,x_W)$, and since there is
    no path from $i$ to $U$ in $G(\tf)$, by~\cref{lemma:paths-I}
    there is a path in $AD(\tf)$ from $\pi(\bar{y}^i)=(y_U,\bar{x}_W^{i})$ to $\pi(\bar{x}^i)=(x_U,\bar{x}_W^i)$, which concludes.
  \end{enumerate}
  \end{proof}

\section{Conclusion}

Boolean networks are frequently used as modelling tools, with associated dynamics often defined under asynchronous updating.
Elimination of variables can be considered to simplify the computational burden~\cite{naldi2009reduction,calzone2010mathematical,saadatpour2010attractor,saadatpour2011dynamical,saadatpour2013reduction,paracha2014formal}.
While preservation of fixed points and of some reachability properties can be shown~\cite{naldi2009reduction},
the asymptotic behaviour of the full and reduced networks can differ.
In this work we gave conditions on the interaction graph that ensure that cyclic attractors are preserved (\cref{thm:pres-attr}),
and presented examples showing the differences in asymptotic behaviour that can arise when these conditions are not satisfied.
In particular, we showed that Boolean networks with very similar interaction graphs,
differing only in a single intermediate in a chain of intermediate variables,
can have different asymptotic behaviours (\cref{thm:chain}).
We also illustrated how the reduction method can be extended to variables that are negatively autoregulated,
and discussed the effects of this elimination on the attractors (\cref{lemma2,thm:fixed-points}).
We showed that a known bound on the number of attractors of asynchronous dynamics~\cite{richard2009positive}
is a corollary of these properties (\cref{thm:bound-attrs}).
The approach presented here broadens the applicability of variable elimination in the investigation of Boolean network dynamics.
Further extensions of this elimination approach to discrete systems with more than two levels will be considered in the future.

\includegraphics{white.png}

\acknowledgements
The authors thank the reviewers for their comments.

\nocite{*}
\bibliographystyle{abbrvnat}
\bibliography{biblio}
\label{sec:biblio}

\end{document}